\documentclass[a4paper]{article}
\usepackage[british]{babel}
\usepackage{ifluatex}
\ifluatex \else \pdfoutput=1 \fi 

\usepackage{amsmath,amsthm,amssymb}
\usepackage{algorithm,algpseudocode}
\usepackage{enumitem} \setlist{noitemsep}
\usepackage[hscale=0.65]{geometry}
\usepackage{microtype}
\usepackage{todonotes} 

\usepackage{tikz}
\usetikzlibrary{quotes,graphs,calc,external}

\ifluatex  
    \usetikzlibrary{graphdrawing} 
    \usegdlibrary{force,layered}
\else
\fi
\tikzset{interval/.style={arrows={|-|}, very thick}, every path/.style={draw, thick}}

\usepackage{xspace}
\usepackage[strict,style=british]{csquotes}
\usepackage{mathtools} 

\DeclarePairedDelimiter\card{\lvert}{\rvert}

\newcommand*{\define}[1]{\textbf{\textcolor{blue}{#1}}}

\newtheorem{theorem}{Theorem}
\newtheorem{lemma}{Lemma}
\newtheorem{corollary}{Corollary}
\theoremstyle{remark}
\newtheorem{remark}{Remark}

\newcommand*{\neighbour}[1]{\operatorname{N}(#1)}
\newcommand*{\child}[1]{\operatorname{N}^{+}(#1)}
\newcommand*{\parent}[1]{\operatorname{N}^{-}(#1)}
\newcommand*{\vertices}[1]{\operatorname{\mathsf V}(#1)}
\newcommand*{\edges}[1]{\operatorname{\mathsf E}(#1)}

\newcommand*{\treewidth}{\operatorname{tw}}
\newcommand*{\size}[1]{\lvert\lvert #1 \rvert\rvert}

\renewcommand*{\emptyset}{\varnothing}
\newcommand*{\units}{\mathbf S^1}
\newcommand*{\reals}{\mathbf R}

\newcommand*{\integers}{\mathbf Z}

\newcommand*{\powerset}[1]{\mathcal P(#1)}
\newcommand*{\disj}{\sqcup}
\newcommand*{\Disj}{\bigsqcup}

\newcommand*{\class}[1]{\ensuremath{\mathbf{#1}}}
\newcommand*{\sharpP}{\class{\#P}}

\newcommand*{\intervals}{\mathcal I}
\newcommand*{\variables}{\mathcal V}
\newcommand*{\clauses}{\mathcal C}

\newcommand*{\bigO}{\mathcal O}

\usepackage[style=alphabetic, isbn=false, giveninits=true, date=year, urldate=ymd,maxbibnames=15]{biblatex}
\usepackage{hyperref}
\usepackage[capitalise]{cleveref}
\DeclareFieldFormat{eprint:hal}{%
  hal\addcolon\space
  \ifhyperref
    {\href{https://hal.archives-ouvertes.fr/#1}{%
       \nolinkurl{#1}%
       }}
    {\nolinkurl{#1}%
    }
}
\renewbibmacro*{doi+eprint+url}{%
\iftoggle{bbx:url}{\iffieldundef{doi}{\usebibmacro{url+urldate}}{}} {}%
\newunit\newblock   \iftoggle{bbx:eprint} {\usebibmacro{eprint}}     {}%
\newunit\newblock   \iftoggle{bbx:doi}     {\printfield{doi}}     {}}

\title{Counting Kernels in Directed Graphs\\with Arbitrary Orientations}

\author{\href{https://orcid.org/0000-0002-5341-1968}{Bruno Jartoux}\thanks{Ben-Gurion University of the Negev, Israel, \href{mailto:jartoux@post.bgu.ac.il}{jartoux@post.bgu.ac.il}.}}

\date{}
\begin{document}

\maketitle

\begin{abstract}
A kernel of a directed graph is a subset of vertices that is both independent and absorbing (every vertex not in the kernel has an out-neighbour in the kernel). 

Not all directed graphs contain kernels, and computing a kernel or deciding that none exist is \class{NP}-complete even on low-degree planar digraphs.
The existing polynomial-time algorithms for this problem are all obtained by restricting both the undirected structure and the edge orientations of the input: for example, to chordal graphs without bidirectional edges (Pass-Lanneau, Igarashi and Meunier, Discrete Appl Math 2020) or to permutation graphs where each clique has a sink (Abbas and Saoula, 4OR 2005). 

By contrast, we count the kernels of a fuzzy circular interval graph in polynomial time, regardless of its edge orientations, and return a kernel when one exists. (Fuzzy circular interval graphs were introduced by Chudnovsky and Seymour in their structure theorem for claw-free graphs.) 

We also consider kernels on cographs, where we establish \class{NP}-hardness in general but linear-time solvability on the subclass of threshold graphs.
\end{abstract}

\section{Introduction}

A \define{digraph} is a simple, loopless, finite graph plus a choice of orientation for each edge, which may be in either direction or in both. A \define{kernel}\footnote{No relation to the homonymous concept in parametrised complexity.} of a digraph is a subset of vertices that is independent and absorbing, that is, $k$ is a kernel iff the vertices with at least one outgoing edge to $k$ are exactly those not in $k$.

\paragraph{Background.} 
Kernels are commonly attributed to \textcite[\S 4.5.3 in][]{vonNeumann--Morgenstern} (as solutions of certain multi-player games), despite similar ideas in early 20th century proto-game theory (Bouton on Nim \cite{Bouton} and Zermelo on chess \cite{Zermelo}). Besides cooperative and combinatorial games \cite{vonNeumann--Morgenstern, Boros--Gurvich-cores, Fraenkel, Duchet--Meyniel-poison}, kernels make various appearances in logic (paradoxes, argumentation, default logic \cite{Dyrkolbotn,DimopoulosMP,Berge--Rao}) and combinatorics (the proof of the Dinitz conjecture \cite{Galvin}). In most of these applications, kernels model cyclic, self-referential behaviour and the complexity arising therefrom.

Kernels are inclusion-maximal independent sets, but not all digraphs have kernels: for example directed odd cycles do not. Given a digraph $D$ and the set $\ker D$ of its kernels, deciding whether $\ker D$ is nonempty, computing $\card{\ker D}$, or building some $k\in\ker D$ are natural computational problems--all difficult. The decision problem is already \class{NP}-complete, even when $D$ is planar with maximum degree 3 \cite{Chvátal,Fraenkel-hardness}, or when $D$ is the line graph of a bipartite graph (with arbitrary edge orientations) \cite{azizStableMatchingUncertain2022}. The search problem is therefore \class{FNP}-complete. Furthermore, the counting problem is \sharpP-complete \cite{SwarczfiterChaty1994}.
The literature on the existence and computation of kernels and related structures is too vast to relate here; recent summaries are found in \cites{Boros--Gurvich-cores}{Pass-LanneauIM}[\S 3.8 in][]{bang-jensenDigraphs2009}. The main insight is that interleaved directed cycles are the source of computational hardness. In fact, a simple inductive argument shows that every acyclic digraph admits one unique kernel.

\paragraph{Kernels in specific digraph classes.} Stable marriages à la Gale--Shapley \cite{Gale--Shapley,Gusfield--Irving} are in bijection with the kernels of a certain (suitably oriented) perfect line graph \cite{Maffray-line}. The remarkable lattice structure of the set of kernels in that setting has led to interest in other classes of claw-free or perfect graphs. In some, various conditions on edge orientations lead to polynomial-time search \cite[e.g.][]{Maffray-line,Abbas--Saoula,Walicki--Sjur-paperA,Pass-LanneauIM}. Forbidding bidirectional edges and forcing every clique to have a sink are the most common examples. (The precise meaning of \enquote{orientation} varies from author to author. In the present article or in \cite{Maffray-line}, it allows for an edge to be simultaneously oriented in both directions; in e.g.\@ \cite{Pass-LanneauIM} it does not.)

In contrast, we would like to solve kernel problems with restrictions on the undirected structure only---not on the edge orientations. To my knowledge, the sole results in this vein are tractability parametrised by tree-width (see \cref{cliquewidth}) and an exponential-time decision algorithm for general digraphs\footnote{The running time is $\bigO(1.427^n)$. A \enquote{naive} approach would generate and test all $\bigO(3^{n/3})$ maximal independent sets \cite{moonCliquesGraphs1965}: $3^{1/3}\simeq 1.44$.} \cite{Walicki--Sjur-paperA}.

For fuzzy circular interval graphs---a claw-free class which includes the proper interval graphs, see \cref{fig:circulararcinclusions} and the definition in \cref{sec:fcig}---we go beyond the decision problem:

\begin{theorem}\label{thm:fcig}
Given an $n$-vertex fuzzy circular interval graph $D$ with arbitrary edge orientations, 
we can compute the size of $\ker D$ and, if nonempty, return one of its elements in time polynomial\footnote{See \cref{runtimes} for a discussion of running times. They are $\bigO(\card{ D}^3\cdot \size D)$ on fuzzy circular interval graphs, $\bigO(\card{ D}^2\cdot \size D)$ on fuzzy linear interval graphs and $\bigO(\card{ D} \cdot \size D)$ on proper interval graphs.} in $n$. 
\end{theorem}

\begin{corollary}\label{main:decision}
The existence of kernels in fuzzy circular interval graphs with arbitrary edge orientations (including bidirectional edges) is in \class P.
\end{corollary}
From \cref{thm:fcig} and the observation that the maximal independent sets of a graph are the kernels of its all-bidirectional orientation, we get:
\begin{corollary}\label{main:mis}
Counting the maximal independent sets of a fuzzy circular interval graph takes polynomial time.
\end{corollary}
Finally \cref{thm:main}, a technical extension of \cref{thm:fcig} to weighted graphs, yields:
\begin{corollary}\label{main:maxis}
Counting the maximum independent sets of a fuzzy circular interval graph takes polynomial time.
\end{corollary}
This contrasts with the situation in various other graph classes, including the (equally claw-free) line graphs, where maximum independent sets are $\sharpP$-complete (\cref{sec:mis}).

We devote \cref{sec:main} to proving \cref{thm:fcig,thm:main} and \cref{sec:remarks} to remarks on approaches parametrised by graph complexity measures, on maximal independent sets and \sharpP-hardness, and on kernels in cographs and threshold graphs. 

A natural question which we leave open is the complexity of the search problem on circular arc graphs, or even only on interval graphs: the case without bidirectional edges is polynomial \cite{Pass-LanneauIM}.

\paragraph{Notation and terminology.} We use standard (di)graph vocabulary (\enquote{independent set}, \enquote{clique}, \enquote{chordal}, etc.) which any recent textbook will cover.

The vertex set and edge set of a digraph $D$ are $\vertices D$ and $\edges D$. We write $\card D \coloneqq \card{\vertices D}$ and $\size D \coloneqq \card{\vertices D} + \card{\edges D}$.

The in- and out-neighbours of a vertex $x$ form the sets $\parent x$ and $\child x$.  Also let $\neighbour x \coloneqq\parent x \cup\child x$. Bidirectional edges may lead to $\parent x \cap\child x \neq\emptyset$. The underlying digraph will always be clear from the context. 
For a set of vertices we also write $\parent S \coloneqq \cup_{x\in S} \parent x$, and so on. Thus a subset of vertices $k$ is a kernel iff $\parent k = \vertices D \setminus k$.

We say that a set of vertices $A$ \define{absorbs} a set of vertices $B$ if $B \subset (A \cup \parent A)$. 

If $S$ is a subset of vertices, $D[S]$ is the sub-digraph induced by $S$. 

Finally, we use $\disj$ in lieu of $\cup$ to denote a union of pairwise disjoint sets.

\section{Counting kernels in fuzzy circular interval graphs}
\label{sec:main}

In (undirected) claw-free graphs, the independent sets have remarkable properties allowing efficient algorithms for \emph{maximum independent set} and related problems \cite{FaenzaOS,hermelinDominationStarsOut2019}. This makes them a popular graph class for algorithmic applications.

In a series of seven papers and a survey, \textcite{Chudnovsky--Seymour-survey} gave a structure theorem for claw-free graphs. Disregarding subclasses with bounded independence number (amenable to brute force as far as kernels are concerned), the \define{quasi-line} graphs are an important stepping stone in this decomposition. They are the graphs in which the neighbourhood of every vertex is the union of two cliques. In turn, connected quasi-line graphs are either \emph{fuzzy circular interval graphs} (see below), or generalised line graphs \cites[\nopp 1.1]{Chudnovsky--Seymour-quasi-line}[\nopp 2.2]{Chudnovsky--Seymour-survey}. Since the existence of kernels is \class{NP}-complete on arbitrarily-oriented line graphs \cite{azizStableMatchingUncertain2022}, we focus on the first class. See \cref{fig:circulararcinclusions} for a diagram depicting the inclusions between all these classes.

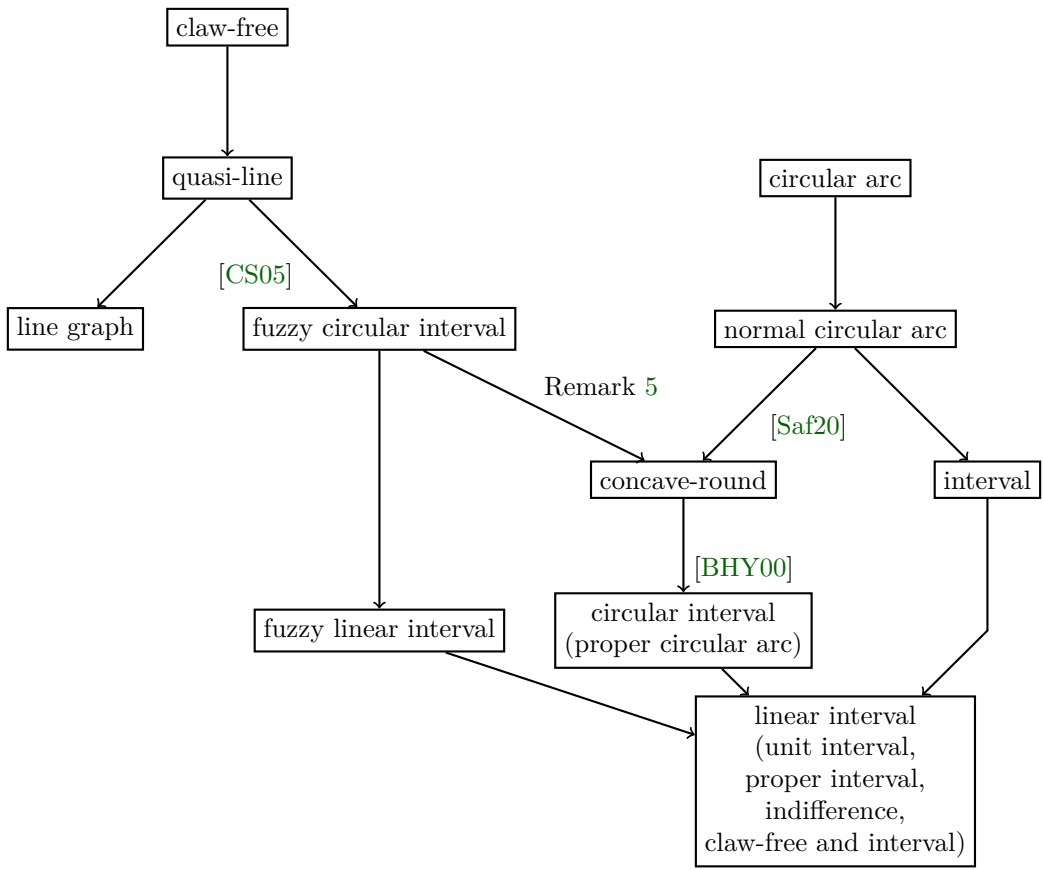
\begin{figure}
    \centering
    \tikzsetnextfilename{inclusions}
    \ifluatex \tikzexternaldisable \else \tikzset{external/remake next=false}  \fi
    \begin{tikzpicture}[align=center]
        \graph [layered layout, level distance=2cm, nodes={draw, rectangle}, sibling distance=40mm] {
        "claw-free" -> "quasi-line" -> "line graph";
        "quasi-line" ->["\cite{Chudnovsky--Seymour-survey}"'] "fuzzy circular interval" -> "fuzzy linear interval" -> "linear interval\\(unit interval,\\proper interval,\\indifference,\\claw-free and interval)";
        "fuzzy circular interval" ->["\Cref{concaveround}"] "concave-round"  ->["\cite{Bang-JensenHY}" near end] "circular interval\\(proper circular arc)"   -> "linear interval\\(unit interval,\\proper interval,\\indifference,\\claw-free and interval)";
        "circular arc" -> "normal circular arc" ->["\cite{safeCharacterizationLineartimeDetection2020}"] "concave-round";
        "interval" -> "linear interval\\(unit interval,\\proper interval,\\indifference,\\claw-free and interval)";
        "normal circular arc" -> "interval";
    };
    \end{tikzpicture}
    \caption{Proper inclusions among some graph classes. Inclusions without a reference are well-known or obvious.
    \Cref{fig:FLIGexample} shows a fuzzy linear interval graph which is not circular arc, and the claw $K_{1,3}$ itself is an interval graph that is not claw-free. The Information System on Graph Classes and their Inclusions is a helpful resource \cite{InformationSystemGraph}.}
    \label{fig:circulararcinclusions}
\end{figure}
\ifluatex \tikzexternalenable \fi

\subsection{Fuzzy circular interval graphs and subclasses}
\label{sec:fcig}

For every two points $a$ and $b$ of the unit circle $\units \subset \reals^2$, define $[a;b]$ as the closed interval of $\units$ joining $a$ to $b$ anticlockwise, with $[a;a]\coloneqq\{a\}$. Also let $[a;b)\coloneqq[a;b]\setminus \{b\}$, and so on.

Following \textcite{Chudnovsky--Seymour-survey}, a graph $G$ is a \define{fuzzy circular interval graph} (FCIG) if it has a model consisting of a function $f\colon \vertices G\to \units$ and a set $\intervals$ of intervals of $\units$ of the form $[a;b]$ with $a\neq b$ such that: 
\begin{enumerate}
    \item no two intervals in $\intervals$ share an endpoint; \label{distinct_endpoints}
    \item there are no proper inclusions among members of $\intervals$;
    \item for every two adjacent vertices $x$ and $y$ there is some interval $I\in \intervals$ such that $f(x),f(y)\in  I$;
    \item for every two non-adjacent vertices $x$ and $y$ and interval $I\in\intervals$, if $f(x),f(y)\in I$ then the endpoints of $I$ are $f(x)$ and $f(y)$. 
\end{enumerate}

\begin{remark}
The model does not fully specify the adjacency relation between vertices it sends to distinct endpoints of the same interval (hence \enquote{fuzzy}). Two FCIGs may admit the same model. For example, consider the vertex set $\{1,2,3,4\}$, and a model having $f(1)=f(2)=a$ and $f(3)=f(4)=b$ with $a,b\in \units$ distinct and $\intervals=\{[a;b]\}$. This model describes several graphs: a cycle, a path, a clique, and the union of two edges.
\end{remark}

The graph is a \define{fuzzy linear interval graph} (FLIG), resp.\@ \define{circular interval graph} (CIG), if it has an FCIG model with $\cup_{J\in \intervals} J \neq \units$, resp.\@ $f$ injective, and a \define{linear interval graph} (LIG) when it has a model where both hold. In other contexts, CIGs are also known as proper circular arc graphs and LIGs as indifference graphs, unit interval graphs or proper interval graphs. See \cref{fig:circulararcinclusions} for some proper inclusions between these and other graph classes. Note that we could equally well define FLIGs by the same conditions as above, except that the codomain of $f$ is the real interval $[0;1]$.

Those four graph classes are each closed under taking induced subgraphs: if $S\subset \vertices G$, the pair $(f\vert_{S}:x\in S \mapsto f(x), \intervals)$ is a model of the subgraph induced by $S$.

\begin{remark}
It is already known that kernel search takes polynomial time on LIGs in which every clique has a sink because LIGs are (strongly) chordal and on CIGs without bidirectional edges because CIGs are circular-arc \cite{Pass-LanneauIM}.
\end{remark}

\paragraph{Nice models.}
Let $G$ be an FCIG (resp.\@ FLIG). An FCIG (FLIG) model $(f\colon \vertices G \to \units,\intervals)$ for $G$ (resp.\@ $f\colon \vertices G \to [0;1])$ is \define{nice} when (i) $\card\intervals \leq \card G$ and (ii), for any two distinct vertices $x$ and $y$, the equality $f(x)=f(y)$ implies that $x$ and $y$ are adjacent.
A nice model can always be computed in time $\bigO(\card{ G}^2\cdot \size G)$. Furthermore, every model in which $f$ already satisfies (ii) can be made nice by deleting some intervals while keeping the same $f$ \cite{OrioloPS}.

\begin{figure}[h]
    \centering
    \tikzsetnextfilename{graph} 
    \ifluatex \tikzset{external/remake next}  \else \tikzset{external/remake next=false} \fi
    \begin{tikzpicture}[every node/.style={circle, draw, double, text height=5pt},every edge/.style={draw, thick}]
        \graph [spring layout] {
            a -- b -- c -- d -- a; c--e; d--e --[orient=0] f};
    \end{tikzpicture}
    \qquad
    \tikzset{external/export next=false}
    \begin{tikzpicture}[every node/.style={below=5pt, align=flush center}]
        \draw[dotted, ->] (0,0)-- (1,0) node {$a$\\$b$} -- (2,0) node {$c$\\$d$} --  (3,0) node {$e$} -- (4,0) node {$f$} -- (5,0);
        \foreach \x in {1,...,4}
            \draw (\x,-.1) -- (\x,0.1);
        \draw[interval] (1,0.5) -- (2,0.5); 
        \draw[interval] (1.5,1) -- (3.5,1);
        \draw[interval] (2.5, 0.5) -- (4.5,0.5);
    \end{tikzpicture}
    \caption{A graph and a nice FLIG model. The reader can check that this graph is not an LIG, nor more generally a circular arc graph. Hint: the disjoint union of a 4-cycle and a vertex is not a circular arc graph.}
    \label{fig:FLIGexample}
    
\end{figure}
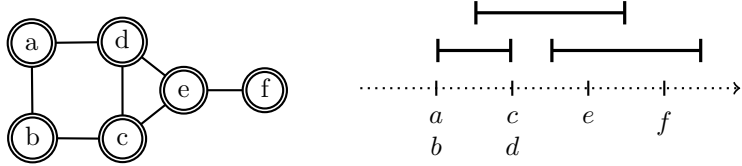

\paragraph{Roadmap.}
In the remainder of this section, we prove \cref{thm:main} in three steps: we obtain a pared-down version for FLIGs via dynamic programming, we extend it to FCIGs, then we discuss how some additional kernel features may be tracked.

\subsection{The structure of kernels in fuzzy linear interval graphs}
\label{flig}

We start with some results about the structure of FLIGs and their kernels. In all that follows, fix an FLIG $D$ and a nice model $(f\colon \vertices D \to [0;1],\intervals)$. We will build a weak ordering of $\vertices D$ whose restriction to each independent set (in particular, to each kernel) is a total order. Intuitively, we will then construct kernels of $D$ by \enquote{guessing} their vertices one after the other \enquote{from right to left}.

\paragraph{Weak vertex ordering.}
 The model gives $\vertices D$ a weak ordering (a \enquote{total order with ties}): for every two vertices $x$ and $y$ put $x\prec y$ iff $f(x)<f(y)$. Since the model is nice, if $x$ and $y$ are non-adjacent then $f(x)\neq f(y)$, so either $x\prec y$ or $y\prec x$. Thus every independent set (in particular every kernel) is totally ordered by $\prec$, even though $\prec$ is in general not total on $\vertices D$. 
 
 Further write $x\preceq y$ iff $y\not\prec x$ (equivalently, iff $f(x)\leq f(y)$). Now we may talk about the vertex intervals 
 \begin{align*}
     [x;y] \coloneqq \{z\in \vertices D \colon x \preceq z \preceq y\},\\
     [x;y) \coloneqq \{z\in \vertices D \colon x \preceq z \prec y\},
 \end{align*}
 and so forth. In particular $[x]\coloneqq[x;x]$ is a clique containing not only $x$ but all vertices mapped to $f(x)$.

For every vertex $x$, the uniqueness of interval endpoints means that there is at most one interval $[a;f(x)]\in \intervals$.  We set $\delta(x)\coloneqq f^{-1}(a)$ if it exists and $\delta(x)\coloneqq\emptyset$ otherwise. When nonempty, $\delta(x)$ is a clique that may contain both neighbours and non-neighbours of $x$. Since two intervals in $\intervals$ may not share an endpoint,
\begin{equation}
    \forall x \in \vertices D, \quad \forall y\in \delta(x), \quad \delta(y)=\emptyset. \label{boundaries}
\end{equation}
    
For every two vertices $x$ and $y$ let the set $K(x,y)$ comprise those kernels of the induced subgraph $D[[x;y]]$ that contain $x$ and $y$:
\[ K(x,y) \coloneqq \{k \in \ker D[[x;y]]\colon x,y\in k\}. \]

In this section we focus on the sets $K(x,y)$, for which we give a recursive description. We motivate this choice by observing that, while in general $K(x,y) \not\subset \ker D$, we do have 
\[ \ker D \subset \Disj_{x,y\in \vertices D} K(x,y), \]
since a kernel with leftmost and rightmost vertices $x$ and $y$ belongs to $K(x,y)$.

\paragraph{Preliminary facts.} Let us establish three properties of $\prec$. These will also hold for induced subgraphs of $D$ when equipped with the same weak order.

\begin{lemma}\label{fencepost}
If $a, b \in \vertices D$ are not adjacent, $a \prec b$  and $c\in\neighbour a$ then $c \preceq b$ (resp.\@ $c\in \neighbour b$ and $a\preceq c$). 
\end{lemma}
\begin{proof}
Let $c\in \neighbour a$: there exists an interval $I\in\intervals$ with $f(a), f(c) \in I$. But then $f(b)$ is not in the interior of $I$, as this would cause $a$ and $b$ to be adjacent. Hence $c \preceq b$. The other inclusion is similar.
\end{proof}

\begin{lemma}\label{localkernel}
If $k\in \ker D$ and $a,b\in k$ with $a\prec b$, then $k\cap [a;b]$ absorbs $(a;b)$.
\end{lemma}

\begin{proof}
Since $k$ is a kernel, every vertex $c\in (a;b)\setminus k$ must have an out-neighbour $d\in k$. If $d \prec a$ or $b \prec d$ then \cref{fencepost} yields $\neighbour d\cap(a;b)=\emptyset$, which contradicts $c\in \neighbour d$. Hence $d \in [a;b]$.
\end{proof}

Intuitively, \cref{localkernel} describes the local structure of kernels in $D$: consecutive kernel vertices are non-adjacent and absorb the open vertex interval which they delimit.

\begin{lemma}\label{secondrightmost}
If $x$ and $y$ are distinct vertices, $k\in K(x,y)$ and $m\in k\setminus\{x,y\}$, then $k\cap [m; y]$ absorbs $(m; y]$, resp.\@ $k\cap [x; m]$ and $[x;m)$.
\end{lemma}
\begin{proof}
From \cref{localkernel} applied to $D[[x;y]]$ (with $a=m$ and $b=y$), $k\cap [m; y]$ absorbs $(m; y)$. The sole remaining case is that of a vertex $z \in [y]\setminus k$. Since $k$ is a kernel, $z$ has an out-neighbour $w\in k$. If $w\in [m; y]$ we are done. Otherwise, it must be that $w \prec  m$, but then \cref{fencepost} (with $a = w$, $b=m$ and $c=z$) gives $z\preceq m$, which contradicts $z\in [y]$.
The second result is symmetrical.
\end{proof}

\paragraph{Recursive description of the kernels.}
Now we may describe the sets of the form $K(x,y)$ with $x\prec y$. Such sets are obtained by \enquote{guessing} their members in decreasing $\prec$-order, ensuring at each step that the newly selected vertex satisfies some compatibility conditions. The main source of complications is the possibility of choosing the next vertex from the set $\delta(y)$ after selecting $y\in \vertices D$. 

To address this issue, we sort the kernels in $K(x,y)$ into $\card{\delta(y)}+1$ disjoint subsets depending on their intersection with $\delta(y)$:
\begin{align}
K_0(x,\emptyset, y) &\coloneqq \{k \in K(x,y) \colon k \cap \delta(y) = \emptyset\},
\intertext{and, for every $m\in\delta(y)$,}
K_0(x,m,y)&\coloneqq \{k \in K(x,y) \colon m\in k\}.
\end{align}

Next, we introduce the sets from which kernel vertices will be selected at each step. For every two vertices $x\prec y$, let 
\begin{align}
\Lambda(x,y) &\coloneqq \{m\in (x;y)\setminus (\neighbour {\{x, y\}} \cup \delta(y))\colon \text{$\{m,y\}$ absorbs $(m;y]$}\}, \label{lambda}
\intertext{and, for every $m \in\delta(y) \setminus \neighbour {\{x, y\}}$,}
\Lambda'(x,m,y) &\coloneqq \{w\in (x;m)\setminus \neighbour {\{x, m\}} \colon \text{$\{w,m,y\}$ absorbs $(w;y]$}\}.\label{lambdaprime}
\end{align}

\begin{remark}
The functions $\Lambda$ and $\Lambda'$ take $\bigO(\card{D}^3)$ values, which we may compute in polynomial time.
\end{remark}

Recall that $\disj$ indicates disjoint unions.

\begin{lemma}
If $x \prec y$ are non-adjacent vertices and $\{x,y\}$ does not absorb $[x;y]$ then
\begin{align}
    K(x,y) &= \Disj_{m\in(\delta(y)\setminus \neighbour {\{x, y\}} \disj\{\emptyset\}
    } K_0(x,m,y), \label{rec1} &\\
    K_0(x,\emptyset,y) &= \hspace{0.9cm} \Disj_{w\in \Lambda(x,y)} \{ k \disj \{y\}\colon k\in K(x,w)\}, &\label{rec2} 
    \shortintertext{and for every $m \in\delta(y) \setminus \neighbour {\{x, y\}}$,}
    K_0(x,m,y) &= 
    \begin{dcases*}
    \hspace{1cm}\left\{\{x,m,y\}\right\} &if $\{x,m,y\}$ absorbs $[x;y]$,\\ 
    \hspace{0.35cm}\Disj_{w\in \Lambda'(x,m, y)} \{ k \disj \{m,y\}\colon  k\in K(x,w)\} &otherwise.
    \end{dcases*} \label{rec4}
\end{align}
Finally if $x\prec y$ and $\{x,y\}$ does absorb $[x;y]$ then $K(x,y)=\{\{x,y\}\}$.
\end{lemma}

\begin{proof}
\Cref{rec1} expresses the fact that a kernel in $K(x,y)$ has at least three vertices, since by hypothesis $\{x,y\}$ is not itself a kernel. The disjunction is on whether the second-rightmost kernel vertex belongs to the clique $\delta(y)$, and in the affirmative on the unique kernel vertex in $\delta(y)$.

Let us prove \cref{rec2}. It is a disjunction with respect to the second-rightmost vertex in the kernel.

For the direct inclusion, 
let $k\in K_0(x,\emptyset,y)$ and let $w$ be the second-rightmost vertex of $k$. 
\begin{itemize}
    \item Observe that $w\notin \neighbour {\{x, y\}}$ since $k$ is an independent set, that $w\notin\delta(y)$ (by hypothesis), and that $\{w,y\}$ absorbs $(w;y]$ by \cref{secondrightmost}. Hence $w\in\Lambda(x,y)$.
    \item It remains to argue that $k\setminus \{y\} \in K(x, w)$. Its independence and the inclusions $\{x, w\}\subset k\setminus\{y\}\subset[x; w]$ are clear. By \cref{secondrightmost}, $k\setminus\{y\}$ absorbs $[x;w)$. We have only to show that every $z\in[w]\setminus \{w\}$ has an out-neighbour in $k\setminus \{y\}$. 
We do know that $z$ has an out-neighbour in $k$, which suffices unless this out-neighbour is $y$.  But then some interval $I\in\intervals$ would contain both $f(w)=f(z)$ and $f(y)$. Because $w$ and $y$ are not adjacent, the endpoints of $I$ would be $f(w)=f(z)$ and $f(y)$, which contradicts $w\notin \delta(y)$.
\end{itemize}

For the reverse inclusion, fix $w\in\Lambda(x,y)$ and $k\in K(x,w)$ and let us show that $k\disj\{y\}\in K_0(x,\emptyset,y)$.
The set $k\disj\{y\}$ is independent: by \cref{fencepost} the only possible exception would be an edge joining $w$ and $y$, but $w\in\Lambda(x,y)$ rules out such an edge. It absorbs $(w;y]$ by definition of $\Lambda(x,y)$ and $[x;w]$ by definition of $K(x,w)$, hence it absorbs $[x;y]$. That it contains $w$ and $y$ is clear.

 \Cref{rec4} is a disjunction with respect to the third rightmost vertex in the kernel. In the edge case where this vertex is $x$, the whole kernel is $\{x,m,y\}$. Consider the other case.
 
 For the direct inclusion, let $k\in K_0(x,m,y)$ and let $w$ be the third rightmost vertex of $k$. Then $w\notin \neighbour{\{x, m\}}$ since $k$ is independent. By \cref{secondrightmost}, $\{w, m, y\}$ absorbs $(w, y]$. Hence $w\in \Lambda'(x,m, y)$. 
 
 It remains to show that $k\setminus\{m, y\} \in K(x,w)$. Its independence is clear. It does absorb $[x;w)$  by \cref{secondrightmost}. The remaining case is that of $z\in[w]\setminus\{w\}$. This $z$ must have an out-neighbour in $k\cap [x;m]$ by \cref{localkernel}. If this out-neighbour were $m$, there would be some $I\in\intervals$ containing $f(w)$ and $f(y)$. Then the independence of $k$ would force $I=[f(w);f(y)]$, in turn leading to the contradiction $w\in \delta(m)= \emptyset$.
 
 For the reverse inclusion, let $w\in\Lambda'(x,m,y)$ and $k\in K(x,w)$. From the definition of $\Lambda'$ it is easy to see that the set $k\disj\{m,y\}$ is independent and absorbs $[x;y]$.
 
 Finally, the case when $x\prec y$ and $\{x,y\}$ does absorb $[x;y]$ is clear.
\end{proof}

\paragraph{The recursive description as a labelled acyclic digraph.}
Together, \cref{rec1,rec2,rec4} let us express $K(x,y)$ as a function of at most $\card{\Lambda(x,y)} + \sum_{m\in\delta(y)} \card{\Lambda'(x,m,y)}$ sets of the form $K(x,y_i)$. We view these results as defining a \emph{subproblem digraph} $\mathcal G(D,\prec)$ whose nodes are pairs $(x,y)\in {\vertices D}^2$ with $x\prec y$, plus a terminal node $F_x$ for each $x\in \vertices D$. If $x$ and $y$ are non-adjacent and $\{x;y\}$ does not absorb $[x;y]$, the node $(x,y)$ of $\mathcal G(D,\prec)$ has the following outgoing edges:
\begin{itemize}
    \item for each $w\in\Lambda(x,y)$, an edge labelled with the symbol $w$ to the node $(x,m)$,
    \item for each $m\in\delta(y)\setminus \neighbour {\{x, y\}}$ and  $w\in \Lambda'(x,m,y)$, an edge labelled with the concatenation $mw$ to the node $(x,w)$, unless $\{x,m,y\}$ absorbs $[x;y]$, in which case the sole outgoing edge is labelled $mx$ and goes to $F_x$.
\end{itemize}
Whereas if $\{x,y\}$ does absorb $[x;y]$ the node $(x,y)$ has an edge labelled $x$ to $F_x$.

Our description of the sets $K(x,y)$ is summarised by the following lemma.

\begin{lemma}
Given an FLIG $D$ and the weak order $\prec$ on its vertex set $\vertices D$ coming from a nice FLIG model, we build a digraph $\mathcal G(D,\prec)$ on the vertex set
\[ \{(x,y)\in (\vertices D)^2 \colon x\prec y\} \disj \{F_x \colon x\in \vertices D\}\]
in time polynomial in $\card{D}$. This digraph is acyclic; it has $\bigO(\card { D}^2)$ nodes and maximum out-degree $\bigO(\card{D})$. Each edge is labelled with one or two vertices of $\vertices D$.

For every pair of vertices $(x,y)$ with $x\prec y$, the set $K(x,y)$ is in bijection with the set of directed paths from the node $(x,y)$ to the node $F_x$. Given one such path, the vertices of the corresponding kernel are read off in decreasing $\prec$-order from the edge labels, omitting $y$.
\end{lemma}

This translates properties of the set $K(x,y)$ such as its cardinality or its smallest elements into well-trodden algorithmic problems on acyclic digraphs (the number of directed paths joining two nodes, the shortest directed paths with edge lengths in $\{1,2\}$).

\subsection{Reducing the circular case to the linear case}
\label{fcigtoflig}

Now we show how the set of all kernels in an FCIG $D$ is expressed in term of the sets $K(x,y)$ in some FLIGs that are all induced subgraphs of $D$. Before we consider FCIGs however, let us mention a simple fact about kernels in general digraphs (proof omitted).
\begin{lemma}\label{delete}
If $D$ is a digraph and $a\in \vertices D$ then
\[ \{ k\in \ker D \colon a\in k\} = \{k\in\ker D[\vertices D\setminus \parent a] \colon a\in k\}. \]
\end{lemma}

Now let $D$ be an FCIG and let $\alpha$ be an arbitrary vertex of $D$. In polynomial time, we compute and fix a nice FCIG model of $D$ \cite{OrioloPS}. We may use the same interval notation as for FLIGs, with the caveat that for every two vertices $a$ and $b$ the vertex intervals $[a;b]$ and $[b; a]$ are generally distinct and intersect in $[a] \cup [b]$.

For each kernel $k\in \ker D$, either $\card k=1$, or there exists a unique pair $a,b\in k$ such that $\alpha \in [a;b)$ and $[a;b]\cap k=\{a,b\}$. That is, setting $K_{a,b}\coloneqq \{k\in \ker D \colon [a;b]\cap k=\{a,b\} \}$,
\[ \ker D = \{ k \in \ker D \colon \card k = 1\} \Disj \left(\Disj_{\alpha \in [a;b)} K_{a,b}\right).\]

Single-vertex kernels can be enumerated by brute force in time $\bigO(\size D)$. Hence we are left to deal with the set $K_{a,b}$ for every two distinct vertices $a$ and $b$. It is not hard to see that:

\begin{lemma} If $a$ and $b$ are not adjacent and $\{a,b\}$ absorbs $[a;b]$, then $K_{a,b}$ is the set of kernels of $D[[b;a]]$ containing $a$ and $b$. If either condition fails, $K_{a,b}=\emptyset$.
\end{lemma}
\begin{proof}
The first part is easy: there can be no other kernel vertex in $(a;b)$ as $[a;b]\subset \neighbour{\{a,b\}}$. The second part involves an FCIG analogue of \cref{localkernel}, proven similarly.
\end{proof}

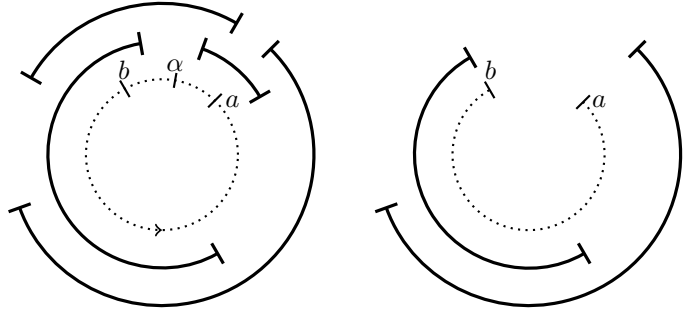
\begin{figure}[ht]
    \centering
    \tikzexternaldisable
    \begin{tikzpicture}[radius=1]
        \draw[dotted] (0,0) circle;
        \path[tips,draw=none, ->] (0:1) arc [start angle=0, end angle=270]; 
        \coordinate  [label=right:$a$] (a) at (45:1);
        \coordinate  [label=above:$b$] (b) at (120:1); 
        \coordinate  [label=above:$\alpha$] (alpha) at (80:1); 
        \draw (80:0.9) -- (80:1.1);
        \path[tips,draw=none, |-|] (a) arc [start angle=45, end angle=120];
        \foreach \source/\target/\height in {30/70/1,60/150/2,100/300/1,-160/45/2}
            {
            \draw[interval]
            let \n1 = {1 +0.5*\height} in
            (\source:\n1) arc [start angle=\source, end angle=\target, radius=\n1];};
    \end{tikzpicture}
    \qquad
    \begin{tikzpicture}[radius=1]
        \coordinate  [label=right:$a$] (a) at (45:1);
        \coordinate  [label=above:$b$] (b) at (120:1); 
        \draw[dotted, |-|] (b) arc [start angle=120, end angle=405];
        \foreach \source/\target/\height in {120/300/1,-160/45/2}
            {
            \draw[interval]
            let \n1 = {1 +0.5*\height} in
            (\source:\n1) arc [start angle=\source, end angle=\target, radius=\n1];};
    \end{tikzpicture}
    \caption{Left: the FCIG $D$ with a nice model. The vertices $a$ and $b$ are not adjacent and absorb $[a;b]$. Right: the subgraph $D'$ obtained by deleting $\parent {\{a,b\}}$ has a nice FLIG model in which $a$ and $b$ are extremal, deduced from the original model by shrinking or deleting intervals without perturbing adjacencies.} 
    \label{fig:FCIGtoFLIG}
\end{figure}

Given $a$ and $b$ non-adjacent and such that $\{a,b\}$ absorbs $[a;b]$, consider the subgraph $D[[b;a]]$ as well as the FCIG interval model $(f,\intervals)$ it inherits from $D$.

By \cref{delete}, we can delete vertices in $\parent {\{a,b\}}$ without modifying the set of kernels containing $a$ and $b$, so we delete $[a;b]\setminus \{a;b\}$. We now have $[a;b]=\{a,b\}$. If the interval $[f(a);f(b)]$ is in $\intervals$, we delete it and retain a valid model, since there is no adjacency it could contribute. Now there are no intervals in $\intervals$ containing both $a$ and $b$. Without removing any edge, we may shrink the intervals in $\intervals$ until some sub-interval of $(f(a);f(b))$ is not covered by $\cup_{I\in\intervals} I$. 

This shows that the resulting induced subgraph $D'$ is an FLIG with the same set of kernels containing both $a$ and $b$ as $D[[b;a]]$, and $D'$ has an FLIG model where $a$ and $b$ are the leftmost and rightmost vertices. This model might not be nice (if we deleted more vertices than intervals), but can be made nice again by deleting unnecessary intervals \cite{OrioloPS}. See \Cref{fig:FCIGtoFLIG} for an illustration.

Thus, computing the cardinality (or an element, or the number of elements of a certain size) of $K_{a,b}$ amounts to computing $K(a,b)$ in the FLIG $D'$ as in \cref{flig}.

Combining the present results with those in \cref{flig}, we obtain \cref{thm:fcig}.

\paragraph{Inapplicability of our approach to general interval graphs.} The proof of \cref{thm:fcig} relies heavily on the structure of the graphs at hand. Namely, the weak vertex ordering arising from a LIG or FLIG model allows for a recursive approach. For example, LIGs can be understood as the intersection graphs of unit intervals--the whole graph structure thus determined by the position of left endpoints. By contrast, there is no such convenient vertex ordering in general interval graphs.

\subsection{Extensions of the algorithm}\label{sub:extensions}
\cref{thm:fcig} can be extended to \cref{thm:main} below.
\begin{theorem}\label{thm:main}
Given a fuzzy circular interval graph $D$ with arbitrary edge orientations, two subsets $S, T \subset \vertices D$, integral vertex weights $w\colon \vertices D \to \integers$, and $t\in \integers$, consider the set
\[ \left\{ k \in \ker D \colon S \subset k \subset T \,\land\, \sum_{x\in k} w(x) = t\right\}.\]
In time polynomial in $\card{D}$, we can compute its number of elements and, if nonempty, return one of them. 
\end{theorem}
In order to do so we modify the graph $\mathcal{G(D,\prec)}$ from \cref{flig} by refining the definitions of $\Lambda$ and $\Lambda'$ as needed, to ensure that kernel vertices respect the newly introduced constraints.

If there are two subsets $S,T\subset \vertices D$ and we want $\mathcal{G(D,\prec)}$ to encode only kernels that include $S$ and are included in $T$, then we may replace \cref{lambda} with
\[
    \Lambda(x,y) \coloneqq \left\{m\in (T \cap (x;y))\setminus (\neighbour {\{x, y\}}\cup \delta(y)) \colon \begin{aligned}&\text{$\{m,y\}$ absorbs $(m;y]$}\\ &[m;y]\cap S \subset \{m; y\}\end{aligned}\right\},\]
and \cref{lambdaprime} similarly. This amounts to deleting some of the edges in $\mathcal{G(D,\prec)}$: the edges with a label not in $T$ and the out-edges from a node $(x,y)$ with at least one label $\ell$ such that $S\cap[\ell;y] \neq \{\ell, y\}$.

If $w\colon \vertices D \to \integers$ is a weight function, then we define the weight of each edge in $\mathcal{G(D,\prec)}$ as the sum of the weights of its labels. For each node $(x,y)$  and terminal $F_x$ of $\mathcal{G(D,\prec)}$, we can list the weights of all directed paths from $(x,y)$ to $F_x$ with multiplicity in a bottom-up fashion, since $\mathcal{G(D,\prec)}$ is acyclic. This still takes time polynomial in $\card D$, independent of $\max_{\vertices D} |w|$.

\subsection{Running times} \label{runtimes}

Let us argue that the running times of our algorithms for counting kernels are
\begin{itemize}
    \item $\bigO(\card{ D}^3\cdot \size D)$ on FCIGs,
    \item $\bigO(\card{ D}^2\cdot \size D)$ on FLIGs and CIGs,
    \item $\bigO(\card{ D} \cdot \size D)$ on LIGs (proper interval graphs).
\end{itemize}

\paragraph{General setting.}
The subproblem digraph $\mathcal G(D,\prec)$ from \cref{flig} can be partitioned into $\card{ D}$ parts, each induced by the vertex set $\{(x,y)\colon y\in\vertices D \land x\prec y\}\disj \{ F_x\}$ for some $x\in \vertices D$. Equivalently, each part correspond to kernels with a same leftmost vertex. There are no edges between parts (and parts need not be connected).  Each part has $\bigO(\card{ D})$ vertices. Hence, for every fixed vertex $x_0$, computing all values of $y\mapsto \card{K(x_0,y)}$ takes time $\bigO(\card{ D}^2)$ once the part containing $F_{x_0}$ has been constructed. Said construction is usually the bottleneck.

For a fixed $y_0 \in \vertices D$, let \begin{align*}
    I &\coloneqq \{ m\prec y_0 \colon m\notin (N(y_0) \cup \delta(y_0)) \land \{m,y_0\} \text{ absorbs } (m;y_0] \}.
    \intertext{For each $x \prec y_0$ the set $\Lambda(x,y_0)$ defined by \cref{lambda} is obtained from $I$ by}
    \Lambda(x,y_0) &= (I \cap (x;y_0))\setminus \neighbour x .
\end{align*}  
Based on this observation, here is a possible algorithm that computes $\Lambda(x,y_0)$ for a fixed $y_0$ and all values $x\in X$, with the vertex $y_0$ and the subset $X \subset \vertices D$ as inputs. Note that, since $\prec$ is not quite a total order but a weak order, it partitions $\vertices D$ into classes ($[x]= \{y\colon y\preceq x \land x \preceq y\}$) and (through quotienting) induces a total order on the set of these classes.

\begin{algorithm}
\caption{An algorithm that computes $\Lambda(x,y_0)$ for a fixed $y_0$ and all values $x\in X$.}
\begin{algorithmic}[1]
\State candidates $\gets \{ m \in \vertices D \setminus (\neighbour y_0 \cup \delta(y_0)) \colon  m \prec y_0\}$ \Comment{as an array of $\card{D}$ bits}
\State $I \gets \emptyset$
\For{each class $C$ in decreasing $\prec$-order starting at $[y_0]$}
\For{$x \in C \cap X$}
\State $\Lambda[x,y_0] \gets I \setminus \neighbour x $ \Comment{$\bigO(\card{ D})$}
\EndFor
    \State{$I \gets I \cup (C\cap \text{candidates})$} \Comment{$\bigO(\card{C})$}
    \For{$m \in C \setminus (\{y_0\} \cup \parent y_0)$}
    \State candidates $\gets \text{candidates} \cap \child m$
    \Comment{$\bigO(\deg m)$}
    \EndFor
\EndFor
\end{algorithmic}
\end{algorithm}

It runs in time $\bigO(\card{D}\cdot \card X + \size D)$. Similarly, if for every $y$ and $m \in\delta(y) \setminus \neighbour y$ we set
\begin{align*}
    J(m,y) &\coloneqq \{w \prec m \colon w \notin \neighbour m \land \text{$\{w,m,y\}$ absorbs $(w;y]$}\},
    \intertext{then, for $x\prec m$, $x\notin\neighbour m$,}
    \Lambda'(x,m,y) &= (J(m,y) \cap(x;m)) \setminus \neighbour x,
\end{align*} 
and a similar algorithm computes the values of $\Lambda'(x,m,y_0)$ for every $x\in X$ and every $m\in \delta(y_0)$ in time $\bigO(\card{\delta(y_0)} \cdot (\card{D}\cdot \card X + \size D))$. Hence, for any fixed $x_0$ it takes time $\bigO(\card{ D}^2\cdot \size D)$ to build only the part of $\mathcal G(D,\prec)$ containing $F_{x_0}$, or $\bigO(\card{ D}^3\cdot \size D)$ to build all parts.

As for FCIGs, each set $K_{a,b}$ from \cref{fcigtoflig} reduces to a set $K(a,b)$ \emph{in an FLIG that potentially depends on both $a$ and $b$}, but careful consideration of the deleted vertices lets us argue that said FLIG only depends on $a$. Thus we have to build $\bigO(\card D)$ parts, for a total running time of $\bigO(\card{ D}^3\cdot \size D)$.

\paragraph{In circular interval graphs.} In circular interval graphs, we can ensure the additional model property that $\delta(y)=\emptyset$ for every vertex $y$. This leads to an empty domain for $\Lambda'$, and therefore to a running time of $\bigO(\card{ D}^2\cdot \size D)$.

\paragraph{In fuzzy linear interval graphs.} We can simplify the analysis by tracking only the rightmost kernel vertex. For example, add an isolated dummy vertex $x_0$ at the left end of the $\prec$-ordering. Because $x_0$ is isolated, it belongs to all kernels. Thus we only need to compute the sets $K(x_0, y)$ where $y\in \vertices D$, that is, we only have to build the part of $\mathcal G(D,\prec)$ that contains $F_{x_0}$. This takes time $\bigO(\card{ D}^2\cdot \size D)$, which is essentially optimal in the sense that it is also the running time of the FLIG recognition algorithm.

\paragraph{In linear interval graphs.} In linear interval graphs, we combine the insights for FLIGs and CIGs: we build a single part of the subproblem digraph and do not compute $\Lambda'$. This leads to a running time  $\bigO(\card{ D} \cdot \size D)$. We leave the matter of getting to e.g.\@ $\bigO(\size D)$ open.

\section{Further remarks}
\label{sec:remarks}

\subsection{Counting maximal and maximum independent sets}
\label{sec:mis}

An independent set of a graph is \define{maximal} when it has no independent superset, and \define{maximum} when its cardinality is largest among independent sets. Clearly, the latter entails the former.

The kernels of a digraph with all edges bidirectional are exactly its maximal independent sets. Hence \cref{thm:main} translates into polynomial-time algorithms for counting maximal independent sets, maximal independent sets of a certain size, and thus also maximum independent sets in FCIGs; see \cref{main:mis,main:maxis}. Compare with chordal graphs and other quasi-line graphs:

\begin{lemma}[\cite{okamotoCountingNumberIndependent2008}]
Counting maximal independent sets is $\sharpP$-complete on chordal graphs.
\end{lemma}

\begin{lemma}
Counting maximum independent sets is $\sharpP$-complete on line graphs of bipartite graphs. 
\end{lemma}

\begin{proof}
Let $G$ be a bipartite graph. In polynomial time, we decide whether $G$ has perfect matchings by building a largest matching (e.g.\@ by the Ford--Fulkerson algorithm) and checking whether it is perfect. In the affirmative, the perfect matchings of $G$ are the maximum independent sets of its line graph. Now an algorithm that counts the latter also counts the former, but counting the perfect matchings of bipartite graphs is $\sharpP$-complete \cite{valiantComplexityComputingPermanent1979}.
\end{proof}

\subsection{Tree-width, clique-width}
\label{cliquewidth}

The tree-width of a digraph $D$ is an integer $\treewidth(D)\in\mathbf N$ measuring its complexity on a scale on which trees are simplest \cite{courcelleGraphStructureMonadic2012}. Note that $\treewidth(D)$ does not depend on edge orientations. 

A set of vertices $k$ being a kernel is expressible in the \emph{monadic second-order logic of directed graphs} (see e.g.\@ \cite{ArnborgLS}), for example by the formula:
\[ \forall v \quad  (v \notin k  \iff \exists u\quad  u\in k \land Evu ), \]
where the relation symbol $E$ encodes directed edges, so the counting version of Courcelle's theorem \cite[Th.\@ 6.56]{courcelleGraphStructureMonadic2012} gives

\begin{theorem}\label{courcelle}
For some computable function $f$ there exist algorithms that, given a digraph $D$, compute $\card{\ker D}$, and, if $\ker D\neq \emptyset$, build a largest and a smallest kernel, all of it in time $\bigO(f(\treewidth(D)) \cdot \size D)$.
\end{theorem}
 The proofs, constructive but quite involved, yield fast growing functions $f$, making said algorithms impractical. Still, this approach has interesting consequences, such as a polynomial-time algorithm for deciding the existence of kernels of size $\bigO(\log^2 \size D)$ in planar digraphs \cite{gutinKernelsPlanarDigraphs2005}.

As for our \cref{thm:main}, there is little overlap with these results since even LIGs have arbitrarily large cliques, hence unbounded tree-width.

\emph{Clique-width}, the other most common complexity measure, is as coarse as tree-width in our orientation-agnostic setting. Indeed, unlike tree-width, clique-width strongly depends on edge orientations: if a class $\mathcal C$ of undirected graphs has unbounded tree-width (say, if $\mathcal C$ contains all complete graphs) then the digraph class obtained by taking all possible edge orientations on members of $\mathcal C$ has unbounded clique-width \cite[Prop.\@ 2.117]{courcelleGraphStructureMonadic2012}. In fact, we show that the existence of kernels is \class{NP}-complete on cographs (which, when undirected, are the graphs of clique-width at most $2$): see \cref{cograph-hardness} below. This contrasts with our \cref{thm:main}, and with the positive results on adequately oriented cographs \cite{Abbas--Saoula}.

Likewise, \emph{bi-min-width} is another complexity measure with applications to kernel-finding and related problems \cite{Jaffke22} which, being unbounded already on orientations of complete graphs, cannot serve our purposes.

\subsection{Cographs and threshold graphs}

One may define \define{cographs} inductively as being either a single-vertex graph, or the union of $G_1$ and $G_2$, or the join\footnote{One obtains the \define{join} of two disjoint undirected graphs by adding every possible edge between them.} of $G_1$ and $G_2$, where $G_1$ and $G_2$ are smaller disjoint cographs. In particular, cliques, independent sets, and matchings are easily identified as cographs.

\begin{theorem}\label{cograph-hardness}
The existence of kernels is \class{NP}-complete on cographs with bidirectional edges.
\end{theorem}

\begin{proof}[Proof of \cref{cograph-hardness}]
The problem is clearly in \class{NP}. 

Fix an instance of the Boolean satisfiability problem (SAT) in conjunctive normal form: a set of variables $\variables$ which defines a set of literals $L = \{v , \bar v\colon v\in \variables\} \simeq \variables \times \{0,1\}$, and a set of disjunctive clauses $\clauses\subset \powerset L$. A solution of the SAT instance is a subset $S\subset L$ such that $v\in S \iff \bar v \notin S$ for each variable $v$ and $S$ intersects every clause in $\clauses$.

Build a digraph on $\clauses\disj L$ by putting a bidirectional edge between every pair of clauses in $\clauses$, and  between every variable $v\in \variables$ and its negation $\bar v$. Also add two special vertices $\alpha$ and $\omega$ with an edge from $\alpha$ to $\omega$. Finally, for each clause $c\in \clauses$ put a directed edge from $c$ to the literals that appear in $c$ and to $\alpha$, as well as directed edges to $c$ from the literals that do not appear in $c$ and from $\omega$. The resulting digraph on $\clauses\disj L \disj \{\alpha,\omega\}$, which we will call $D$, has the following properties:
\begin{itemize}
    \item $D$ has $\card{\variables} + \card{\clauses} + 2$ vertices;
    \item (the undirected graph underlying) $D$ is a cograph, as the join of a complete graph on $\clauses$ and a perfect matching on $L \disj \{\alpha,\omega\}$;  
    \item If $k$ is a kernel of $D$ then $\omega \in k$ and $k\subset L \disj\{\omega\}$. Indeed, since $\child \alpha=\{\omega\}$, either $\alpha\in k$ or $\omega\in k$, but in the first case $\child \omega = \clauses \subset \neighbour \alpha \subset \neighbour k$ contradicts the independence of $k$, since $\omega$ must have an out-neighbour in $k$. Hence $\omega \in k$, which implies $(\clauses\disj\{\alpha\}) \cap k = \neighbour \omega\cap k = \emptyset$;
    \item In view of the previous fact, kernels are subsets of literals (plus $\omega$) containing exactly one of $v$ and $\bar v$ for each variable $v\in \variables$, and such that each $c\in \clauses$ has at least one out-neighbour in the kernel. Thus the kernels of $D$ are in bijection with the solutions of the SAT instance. 
\end{itemize}

Hence this construction is a polynomial-time parsimonious reduction of SAT to the existence of kernels in cographs.
\end{proof}

We conclude by noting that the counting problem is not only polynomial but linear on a particular subclass of cographs, even though said subclass, having unbounded tree-width, does not fall under the purview of \cref{courcelle}. 

A cograph is a \define{threshold graph} if it has a cograph construction sequence as above where $G_1$ is a single vertex at each step.

\begin{theorem}\label{threshold}
Counting the kernels of an arbitrarily oriented threshold graph $T$ and, when applicable, building one, takes linear time $\bigO(\size T)$. If $\omega(T)$ is the clique number, enumerating its kernels takes time $\bigO(\omega(T) \cdot \card T + \size T)$.
\end{theorem}

These results are already known for the maximal independent sets in undirected threshold graphs \cite{gurskiCountingEnumeratingIndependent2020}.

\begin{proof}[Proof of \cref{threshold}]
For every set $\mathcal F\subset \vertices T$, let $\ker_{\mathcal F} T \coloneqq \{ k \in \ker T \colon k \cap \mathcal F = \emptyset\}$. A threshold graph construction sequence for $T$ is found in linear time \cite[\S 1.4.2]{mahadevThresholdGraphsRelated1995}, so we give a recursive description of $\ker_{\mathcal F} T$. 
\begin{itemize}
    \item If $T$ is a single vertex $v$, then \[ \ker_{\mathcal F} T = \begin{dcases*} \emptyset &if $v \in\mathcal F$,\\ \{\{v\}\} &otherwise.\end{dcases*}\]
    \item If $T$ is the disjoint union of a vertex $v$ and a subgraph $T'$, then 
    \[ \ker_{\mathcal F} T =\begin{dcases*} \emptyset &if $v\in\mathcal F$,\\ \{\{v\}\disj k \colon k \in \ker_{\mathcal F} T'\} &otherwise.\end{dcases*}\]

    \item If $T$ is the join of a vertex $v$ and a subgraph $T'$, then 
    $\{v\}\in \ker_{\mathcal F} T$ iff $v$ absorbs $T'$ and $v \notin \mathcal F$ (and we can check these conditions in time $\bigO(\deg v + n)$) and regardless 
    \[\ker_{\mathcal F} T\setminus\{\{v\}\} =\ker_{\mathcal F \setminus \{v\}} T'
    .\]
\end{itemize}
This lets us compute $\card{\ker_{\mathcal F} T}$ in time $R(\card T, \card{\edges T})$ where (for some absolute constant $C$) the three cases above give
\begin{align*}
    R(n, m) &\leq C + R(n - 1, m) + \max_{v\in \vertices T} (C \cdot (\deg v + n) + R(n -1 , m -\deg v))\\
    &= \bigO(m + n).
\end{align*} 

For enumerating $\ker_\mathcal{F} T$ the same reasoning applies, except that each \enquote{disjoint union} step spends constant time appending a vertex to each set in $\ker_\mathcal{F} T'$. As a threshold graph, $T$ has $\omega(T)$ maximal independent sets \cite{gurskiCountingEnumeratingIndependent2020}, so $\card{\ker_\mathcal{F} T'} \leq \card{\ker_\mathcal{F} T} \leq \omega(T)$ and we obtain  
\begin{align*}
    R'(n, m) &\leq C + (R'(n - 1, m) + C \cdot \omega(T)) \\& \quad + \max_{v\in \vertices T} (C \cdot (\deg v + n) + R'(n -1 , m -\deg v)  )\\
    &= \bigO(\omega(T) \cdot n + m).\qedhere
\end{align*} 
\end{proof}

\begin{remark}
The maximal independent sets of a graph are enumerable with polynomial delay, hence, if there are $\mu$ such sets, enumerable in total time $\mu \cdot n^{\bigO(1)}$. This is unlikely to extend to kernels. Even if only on cographs, a polynomial-delay algorithm for kernels would solve SAT and (because of the parsimonious reduction) an enumeration algorithm polynomial in the total size of its output would solve unambiguous-SAT. This would respectively establish $\class P = \class{NP}$ and, by the Valiant--Vazirani theorem, $\class{RP} = \class{NP}$ \cite{valiantNPEasyDetecting1986}.
\end{remark}
\subsection{Other graph classes}

\paragraph{Concave-round graphs.}
A graph is \define{concave-round} if its vertices have a circular ordering in which every \emph{closed} neighbourhood (that is, $\{x\}\disj \neighbour x$) is an interval \cite{Bang-JensenHY}. Concave-round graphs are quasi-line. They form a subclass of normal circular-arc graphs \cite{safeCharacterizationLineartimeDetection2020} and a proper superclass of CIGs. Unknowingly, we have already solved the problem for this class, because of the following easy but unpublished inclusion.

\begin{remark}\label{concaveround}
Concave-round graphs form a proper subclass of the fuzzy circular interval graphs. More precisely, they are either fuzzy linear interval graphs or circular interval graphs (proper circular arc graphs).
\end{remark}
\begin{proof}
A concave-round graph is (at least) one of circular interval and co-bipartite \cite{tuckerMatrixCharacterizationsCirculararc1971,safeCharacterizationLineartimeDetection2020}. In turn, co-bipartite graphs (i.e.\@ graphs with a partition into two cliques) are clearly FLIGs: they admit a single-interval model, with each vertex mapped to one endpoint.
\end{proof}

\paragraph{Interval graphs and interval digraphs.}
There are some results on kernel problems in \emph{interval digraphs} \cite{francisKernelRelatedProblems2021a}. However, said interval digraphs are quite different from arbitrarily-directed interval graphs: interval graphs are chordal, whereas every directed cycle is an interval digraph in the sense of \citeauthor{francisKernelRelatedProblems2021a}. Hence there is little overlap with our results.

\section*{Acknowledgements}
\addcontentsline{toc}{section}{Acknowledgements}
Thanks to F.\@ Meunier for drawing my interest to kernels a few years ago and to Y.\@ Yudistky for her comments on an earlier draft.

 Research supported by the European Research Council under the European Union’s Horizon 2020 research and innovation programme (Grant agreement No. 678765) and by the Israel Science Foundation under Grants 1065/20 and 2891/21.

\phantomsection
\printbibliography[heading=bibintoc]
\end{document}